\documentclass[12pt]{article}

\usepackage{a4wide}
\usepackage[english]{babel}
\usepackage{graphicx}
\usepackage{amsmath}
\usepackage{amssymb}
\usepackage{wrapfig}
\usepackage{xspace}
\usepackage{enumerate}
\usepackage{array}
\usepackage{longtable}
\usepackage{caption}
\usepackage{microtype}

\newtheorem{defin}{Definition}
 
\newtheorem{theo}[defin]{Theorem}
 \newenvironment{theorem}{\begin{theo} \sl}{\end{theo}}
\newtheorem{lem}[defin]{Lemma}
 \newenvironment{lemma}{\begin{lem} \sl}{\end{lem}}
\newtheorem{coro}[defin]{Corollary}
 \newenvironment{corollary}{\begin{coro} \sl}{\end{coro}}
\newtheorem{prop}[defin]{Proposition}
 
\newtheorem{obs}[defin]{Observation}
 \newenvironment{observation}{\begin{obs} \sl}{\end{obs}}

\newenvironment{proof}{\emph{Proof.}}{\hfill $\Box$\\}

\graphicspath{{./images/}}

\newcommand{\etal}{\emph{et al.}\xspace}

\newcommand{\delGraph}{constrained generalized Delaunay graph\xspace}

\newcommand{\C}[2]{\ensuremath{C(#1,#2)}}
\renewcommand{\circle}{convex shape\xspace}
\newcommand{\Vis}{\mathord{\it Vis}}
\newcommand{\reg}[3]{\ensuremath{#1^{#2}_{#3}}}

\title{Constrained Generalized Delaunay Graphs Are Plane Spanners%
  \thanks{Research supported in part by FQRNT, NSERC, Carleton University's President's 2010 Doctoral Fellowship, and JST ERATO Grant Number JPMJER1201, Japan.}
  \thanks{Extended abstracts containing some of the results in this paper appeared in the 27th Canadian Conference on Computational Geometry (CCCG 2015)~\cite{BCR15Rectangle} and in Computational Intelligence in Information Systems (CIIS 2016)~\cite{BCR2016GeneralizedDelaunay}.}}

\author{%
  Prosenjit~Bose,%
    \thanks{School of Computer Science, Carleton University, Ottawa, Canada, 
    \texttt{jit@scs.carleton.ca}}\, 
  Jean-Lou~De~Carufel,%
    \thanks{School of Electrical Engineering and Computer Science, University of Ottawa, Ottawa, Canada, 
    \texttt{jdecaruf@uottawa.ca}}\, 
  and Andr\'e~van~Renssen%
    \thanks{School of Information Technologies, University of Sydney, Sydney, Australia,\newline
    \texttt{andre.vanrenssen@sydney.edu.au}}
}

\date{}

\begin{document}

\maketitle

\begin{abstract}
  We look at generalized Delaunay graphs in the constrained setting by introducing line segments which the edges of the graph are not allowed to cross. Given an arbitrary convex shape $C$, a constrained Delaunay graph is constructed by adding an edge between two vertices $p$ and $q$ if and only if there exists a homothet of $C$ with $p$ and $q$ on its boundary that does not contain any other vertices visible to $p$ and~$q$. We show that, regardless of the convex shape $C$ used to construct the constrained Delaunay graph, there exists a constant $t$ (that depends on~$C$) such that it is a plane $t$-spanner of the visibility graph. Furthermore, we reduce the upper bound on the spanning ratio for the special case where the empty convex shape is an arbitrary rectangle to $\sqrt{2} \cdot \left( 2 l/s + 1 \right)$, where $l$ and $s$ are the length of the long and short side of the rectangle. 
\end{abstract}

\section{Introduction}
A geometric graph $G$ is a graph whose vertices are points in the plane and whose edges are line segments between pairs of vertices. Every edge in a geometric graph is weighted by the Euclidean distance between its endpoints. A graph $G$ is called plane if no two edges intersect properly. The distance between two vertices $u$ and $v$ in $G$, denoted by $\delta_G(u, v)$, or simply $\delta(u, v)$ when $G$ is clear from the context, is defined as the sum of the weights of the edges along a minimum-weight path between $u$ and $v$ in $G$. A subgraph $H$ of $G$ is a $t$-spanner of $G$ (for $t\geq 1$) if for each pair of vertices $u$ and $v$, $\delta_H(u, v) \leq t \cdot \delta_G(u, v)$. The smallest value $t$ for which $H$ is a $t$-spanner is the \emph{spanning ratio} or \emph{stretch factor} of $H$. The graph $G$ is referred to as the {\em underlying graph} of $H$. The spanning properties of various geometric graphs have been studied extensively in the literature (see \cite{BS11,NS-GSN-06} for an overview of the topic). 

Most of the research has focused on constructing spanners where the underlying graph is the complete Euclidean geometric graph. We study this problem in a more general setting with the introduction of line segment {\em constraints}. Specifically, let $P$ be a set of points in the plane and let $S$ be a set of line segments with endpoints in $P$, with no two line segments intersecting properly. The line segments of $S$ are called \emph{constraints}. Two points $u$ and $v$ can \textit{see each other} or \emph{are visible to each other} if and only if either the line segment $u v$ does not properly intersect any constraint (i.e., does not intersect the interior of a constraint) or $u v$ is itself a constraint. If two points $u$ and $v$ can see each other, the line segment $u v$ is a \emph{visibility edge}. The \emph{visibility graph} of $P$ with respect to a set of constraints $S$, denoted $\Vis(P,S)$, has $P$ as vertex set and all visibility edges as edge set. In other words, it is the complete graph on $P$ minus all edges that properly intersect one or more constraints in~$S$.

Visibility graphs have been studied extensively within the context of motion planning amid obstacles. Clarkson~\cite{C87} was one of the first to study spanners in the presence of constraints and showed how to construct a linear-sized $(1+\epsilon)$-spanner of $\Vis(P,S)$. Subsequently, Das~\cite{D97} showed how to construct a spanner of $\Vis(P,S)$ with constant spanning ratio and constant degree. Bose and Keil~\cite{BK06} showed that the Constrained Delaunay Triangulation is a $4 \pi \sqrt{3}/9 \approx 2.42$-spanner of $\Vis(P,S)$. The constrained Delaunay graph where the empty \circle is an equilateral triangle was shown to be a 2-spanner of $\Vis(P,S)$~\cite{BFRV12}. We look at the constrained generalized Delaunay graph, where the empty \circle can be any convex shape. 

In the unconstrained setting, it is known that generalized Delaunay graphs are spanners~\cite{BCCS10}, regardless of the \circle used to construct them. A geometric graph $G$ is a spanner when it satisfies the following properties (defined in Section~\ref{sec:SpanningRatioDelaunay}): it is plane, it satisfies the $\alpha$-diamond property, the spanning ratio of any one-sided path is at most $\kappa$, and it satisfies the visible-pair $\kappa'$-spanner property. In particular, $G$ is a $t$-spanner for $t = 2 \kappa \kappa' \cdot \max \left( \frac{3}{\sin(\alpha/2)}, \kappa \right)$. This upper bound is very general, but unfortunately not tight. 

In special cases, better bounds are known. For example, when the empty \circle is a circle, Dobkin~\etal~\cite{DFS90} showed that the spanning ratio is at most $\pi (1 + \sqrt{5}) / 2 \approx 5.09$. Improving on this, Keil and Gutwin~\cite{KG92} reduced the spanning ratio to $4 \pi / 3\sqrt{3} \approx 2.42$. Recently, Xia showed that the spanning ratio is at most 1.998~\cite{X13}. We note that although Xia's proof is in the unconstrained setting, it still holds in the constrained setting. His proof is based on bounding the length of each edge on the path from a vertex $s$ to $t$ that does not intersect $st$ with the arc of the empty circle defining the edge. The length of edges that cross $st$ is then bounded in terms of the non-crossing edges. In the constrained setting, since the edges that do not cross $st$ are still bounded by arcs of circles that are empty of visible points, his result holds.

Lower bounds are also studied for this problem. Bose~\etal~\cite{BDLSV11} showed a lower bound of 1.58, which is greater than $\pi/2$, which was conjectured to be the tight spanning ratio up to that point. Later, Xia and Zhang~\cite{XZ11} improved this to 1.59.

Chew~\cite{C89} showed that if an equilateral triangle is used instead, the spanning ratio is 2 and this ratio is tight. In the case of squares, Chew~\cite{C86} showed that the spanning ratio is at most $\sqrt{10} \approx 3.16$. This was later improved by Bonichon~\etal~\cite{BGHP12}, who showed a tight spanning ratio of $\sqrt{4 + 2 \sqrt{2}} \approx 2.61$. 

In this paper, we show that the \delGraph $G$ is a spanner whose spanning ratio depends solely on the properties of the empty \circle $C$ used to create it: We show that $G$ satisfies the $\alpha_C$-diamond property and the visible-pair $\kappa_C$-spanner property (defined in Section~\ref{sec:SpanningRatioDelaunay}), which implies that it is a $t$-spanner of $\Vis(P,S)$ for: 
\begin{equation*}
  t = 
  \begin{cases}
    2 \kappa_C \cdot \max \left( \frac{3}{\sin(\alpha_C/2)}, \kappa_C \right), & \text{ \emph{if $G$ is a triangulation}} \\
    2 \kappa_C^2 \cdot \max \left( \frac{3}{\sin(\alpha_C/2)}, \kappa_C \right), & \text{ \emph{otherwise.}}
  \end{cases}
\end{equation*}

This proof is not a straightforward adaptation from the work by Bose~\etal~\cite{BCCS10} due to the presence of constraints. For example, showing that a region contains no vertices that are visible to some specific vertex $v$ requires more work than showing that this same region contains no vertices, since we allow vertices in the region that are not visible to $v$. Also, since the spanning ratio between some pairs of non-visible vertices of the constrained Delaunay graph may be unbounded (i.e., the length of the path between any two non-visible points can be made arbitrarily large by extending a constraint that blocks visibility), any proof of bounded spanning ratio needs to be restricted to the visible pairs of vertices. This implies that induction can only be applied to pairs of visible vertices, meaning that the inductive arguments cannot be applied in a straightforward manner as in the unconstrained case, since in the unconstrained case there is a spanning path between every pair of vertices. 

Our spanning proof works directly on the Delaunay graph, instead of constructing the required paths using the Voronoi diagram as was done in~\cite{BCCS10}. This simplifies the algorithm for constructing these short paths, and also simplifies the proofs.

It is also worth noting that our definition of constrained Delaunay graph is slightly more general than the standard definition of these graphs: While it is usually assumed that all constraints are edges in the graphs, we do not require this and only add a constraint as an edge if it also satisfies the empty circle property used to construct the rest of the graph. Therefore, our result is slightly more general since we show that a subgraph of the standard constrained Delaunay graph is a plane spanner. We elaborate on this point in more detail in Section~\ref{sec:Preliminaries}.

Finally, though the aforementioned result is very general, since it holds for arbitrary \circle{s}, its implied spanning ratio is far from tight. To improve on this, in Section~\ref{sec:rectangles} we consider the special case where the empty \circle $C$ is a rectangle and show that it has spanning ratio at most $\sqrt{2} \cdot \left( 2 l/s + 1 \right)$, where $l$ and $s$ are the length of the long and short side of $C$. This reduces the dependency on the aspect ratio from cubic (as implied by our general bound) to linear.

\section{Preliminaries}
\label{sec:Preliminaries}
Throughout this paper, we fix a bounded \circle $C$. We assume without loss of generality that the origin lies in the interior of $C$. A \emph{homothet} of $C$ is obtained by scaling $C$ with respect to the origin, followed by a translation. Thus, a homothet of $C$ can be written as \[x + \lambda C = \{ x + \lambda z : z \in C \},\] for some scaling factor $\lambda > 0$ and some point $x$ in the interior of $C$ after translation. 

For a given set of vertices $P$ and a set of constraints $S$, the \delGraph is usually defined as follows. Given any two visible vertices $p$ and $q$, let \C{p}{q} be any homothet of $C$ with $p$ and $q$ on its boundary. The \delGraph contains an edge between $p$ and $q$ if and only if $p q$ is a constraint or there exists a \C{p}{q} such that there are no vertices of $P$ in the interior of \C{p}{q} visible to both $p$ and $q$. We assume that no four vertices lie on the boundary of any homothet of $C$. In addition, if $C$ has any straightline segments on its boundary, we assume that no three points lie on a line parallel to any such segment. Like in the unconstrained setting, these assumptions are required to guarantee planarity of the constructed graphs. While it is possible to remove these assumptions and consider the planar subgraphs that contains exactly one of the crossing edges, this significantly complicates the proofs. 

Now, we slightly modify this definition such that there is an edge between two visible vertices $p$ and $q$ if and only if there exists a \C{p}{q} such that there are no vertices of $P$ in the interior of \C{p}{q} visible to both $p$ and $q$.  Note that this modified definition implies that constraints are not necessarily edges of the graph, since constraints may not necessarily adhere to the visibility property. Our modified graph is always a subgraph of the \delGraph. Therefore, any result proven on our modified graph also holds for the graph that includes all the constraints. As such, we prove the stronger result on our modified graph. For simplicity, in the remainder of the paper, when we refer to the \delGraph, we mean our modified subgraph of the \delGraph.

\subsection{Auxiliary Lemmas}
Next, we present three auxiliary lemmas that are needed to prove our main results. First, we reformulate a lemma that appears in~\cite{S03}. 

\begin{lemma}
  \label{lem:HomothetIntersection}
  Let $C$ be a closed convex curve in the plane. The intersection of two distinct homothets of $C$ is the union of at most two sets, each of which is either a segment or a single point. 
\end{lemma}

Though the following lemma (see also Figure~\ref{fig:VisiblePointInsideTriangle}) was applied to constrained $\theta$-graphs in~\cite{BFRV12}, the property holds for any visibility graph. We say that a region $R$ \emph{contains} a vertex $v$ if $v$ lies in the interior or on the boundary of $R$. We call a region \emph{empty} if it does not contain any vertex of $P$ in its interior. We also note that we distinguish between \emph{vertices} and \emph{points}. A \emph{point} is any point in $\mathbb{R}^2$, while a \emph{vertex} is part of the input. 

\begin{lemma}
  \label{lem:ConvexChain}
  Let $u$, $v$, and $w$ be three arbitrary points in the plane such that $u w$ and $v w$ are visibility edges and $w$ is not the endpoint of a constraint intersecting the interior of triangle $u v w$. Then there exists a convex chain of visibility edges from $u$ to $v$ in triangle $u v w$, such that the polygon defined by $u w$, $w v$ and the convex chain is empty and does not contain any constraints.
\end{lemma}
\begin{figure}[ht]
  \begin{center}
    \includegraphics{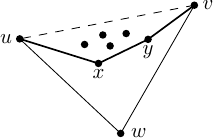}
  \end{center}
  \caption{The convex chain between vertices $u$ and $v$, where thick lines are visibility edges.}
  \label{fig:VisiblePointInsideTriangle}
\end{figure}

Let $p$ and $q$ be two vertices that can see each other and recall that \C{p}{q} is a homothet of $C$ with $p$ and $q$ on its boundary. Extend to half-lines with source $p$ all constraints and edges that have $p$ as an endpoint and that intersect \C{p}{q} (see Figure~\ref{fig:Region}a). Define the clockwise neighbor of $p q$ to be the half-line that minimizes the strictly positive clockwise angle with $p q$ (or the tangent of \C{p}{q} at $p$ if this neighbor does not exist) and define the counterclockwise neighbor of $p q$ to be the half-line that minimizes the strictly positive counterclockwise angle with $p q$ (or the tangent of \C{p}{q} at $p$ if this neighbor does not exist). We define the \emph{cone} $\reg{C}{p}{q}$ that contains $q$ to be the region between the clockwise and counterclockwise neighbor of $p q$. Finally, let \reg{\C{p}{q}}{p}{q}, the \emph{region of \C{p}{q} that contains $q$ with respect to $p$}, be the intersection of \C{p}{q} and $\reg{C}{p}{q}$ (see Figure~\ref{fig:Region}b). 

\begin{figure}[ht]
  \begin{center}
    \includegraphics{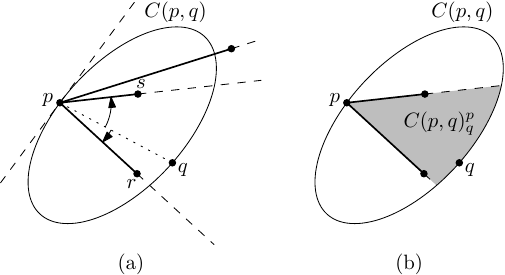}
  \end{center}
  \caption{Defining the region of \C{p}{q} that contains $q$ with respect to $p$: (a) The clockwise and counterclockwise neighbor of $p q$ are the half-lines through $p r$ and $p s$, (b) \reg{\C{p}{q}}{p}{q} is shaded gray.}
  \label{fig:Region}
\end{figure}

\begin{lemma}
  \label{lem:VisibleVertex}
  Let $p$ and $q$ be two vertices that can see each other and let \C{p}{q} be any convex shape with $p$ and $q$ on its boundary. If there is a vertex $x$ in \reg{\C{p}{q}}{p}{q} (other than $p$ and $q$) that is visible to $p$, then there is a vertex $y$ (other than $p$ and $q$) that is visible to both $p$ and $q$ and such that triangle $p y q$ is empty. 
\end{lemma}
\begin{proof}
  We have two visibility edges, namely $p q$ and $p x$. Since $x$ lies in \reg{\C{p}{q}}{p}{q}, $p$ is not the endpoint of a constraint such that $q$ and $x$ lie on opposite sides of the line through this constraint. Hence, we can apply Lemma~\ref{lem:ConvexChain} and we obtain a convex chain of visibility edges from $x$ to $q$ and the polygon defined by $p q$, $p x$ and the convex chain is empty and does not contain any constraints. Furthermore, since the convex chain is contained in triangle $p x q$, which in turn is contained in \C{p}{q}, every vertex along the convex chain is contained in \C{p}{q} (see Figure~\ref{fig:VisibleToBoth}). 

  \begin{figure}[ht]
    \begin{center}
      \includegraphics{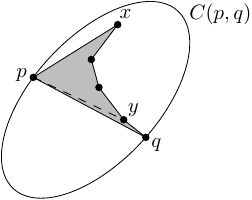}
    \end{center}
    \caption{Vertex $y$ lies in \C{p}{q} and is visible to both $p$ and $q$.}
    \label{fig:VisibleToBoth}
  \end{figure}

  Let $y$ be the neighbor of $q$ along this convex chain. Hence, $y$ is visible to $q$ and contained in \C{p}{q}. Furthermore, $p$ can see $y$, since the line segment $p y$ is contained in the polygon defined by $p q$, $p x$ and the convex chain, which is empty and does not contain any constraints. This also implies that triangle $p y q$ is empty. 
\end{proof}

\section{The Constrained Generalized Delaunay Graph Is Plane with Constant Spanning Ratio}
Before we show that every \delGraph is a spanner, we first show that they are plane. 

\subsection{Planarity}
\label{subsec:planarity}
In order to show that the \delGraph is plane, we first observe that no edge $p q$ of the graph can contain a vertex in its interior, as this vertex would lie in \C{p}{q} and be visible to both endpoints of the edge, contradicting the existence of the edge $p q$. Let $\partial C$ denote the boundary of $C$. 

\begin{observation}
  \label{lem:NoExtension}
  Let $p q$ be an edge of the \delGraph. The line segment $p q$ does not contain any vertices other than $p$ and $q$. 
\end{observation}

\begin{lemma}
  \label{lem:PlaneDelaunay}
  The \delGraph is plane. 
\end{lemma}
\begin{proof}
  We prove this by contradiction, so assume that there exist two edges $p q$ and $r s$ that intersect. It follows from Observation~\ref{lem:NoExtension} that neither $p$ nor $q$ lies on $r s$ and that neither $r$ nor $s$ lies on $p q$, so the edges properly intersect. Since $p q$ is contained in \C{p}{q} and $r s$ is contained in \C{r}{s}, $\partial\C{p}{q}$ and $\partial\C{r}{s}$ intersect or one of \C{p}{q} and \C{r}{s} contains the other. 

  \begin{figure}[ht]
    \begin{center}
      \includegraphics{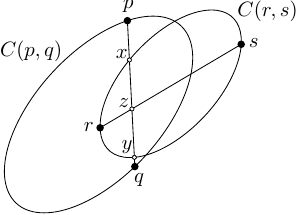}
    \end{center}
    \caption{$\partial\C{p}{q}$ and $\partial\C{r}{s}$ intersect and $p q$ intersects \C{r}{s} at $x$ and $y$.}
    \label{fig:DelaunayPlane}
  \end{figure}

  We first show that this implies that $p$ or $q$ lies in \C{r}{s} or $r$ or $s$ lies in \C{p}{q}. If  one of \C{p}{q} and \C{r}{s} contains the other, this holds trivially. If the two homothets intersect and either $p \in \C{r}{s}$ or $q \in \C{r}{s}$, we are done, so assume that neither $p$ nor $q$ lies in \C{r}{s}. Lemma~\ref{lem:HomothetIntersection} states that $\partial\C{p}{q}$ and $\partial\C{r}{s}$ intersect at most twice. These intersections split $\partial\C{p}{q}$ into two parts: one that is contained in \C{r}{s} and one that is not. Since $p \not \in \C{r}{s}$ and $q \not \in \C{r}{s}$, $p$ and $q$ lie on the arc of $\partial\C{p}{q}$ that is not contained in \C{r}{s} (see Figure~\ref{fig:DelaunayPlane}). However, $p q$ intersects \C{r}{s}, since otherwise $p q$ cannot intersect $r s$. Let $x$ and $y$ be the two intersections of $p q$ with $\partial\C{r}{s}$ (if $\partial\C{r}{s}$ is parallel to $p q$, $x$ and $y$ are the two endpoints of the interval of this intersection). We note that $x$ and $y$ split $\partial\C{r}{s}$ into two parts, one of which is contained in \C{p}{q}, and $r$ and $s$ must lie in different parts. In particular, one of $r$ and $s$ lies on the part that is contained in \C{p}{q}, proving that $r \in \C{p}{q}$ or $s \in \C{p}{q}$. This proves that $p \in \C{r}{s}$, $q \in \C{r}{s}$, $r \in \C{p}{q}$, or $s \in \C{p}{q}$. 
  
  In the remainder of the proof, we assume without loss of generality that $r \in \C{p}{q}$ (see Figure~\ref{fig:DelaunayPlane}). Let $z$ be the intersection of $p q$ and $r s$. Hence, $z$ can see both $p$ and $r$. Also, $z$ is not the endpoint of a constraint intersecting the interior of triangle $p z r$. Therefore, it follows from Lemma~\ref{lem:ConvexChain} that there exists a convex chain of visibility edges from $p$ to $r$. Let $v$ be the neighbor of $p$ along this convex chain. Since $v$ is part of the convex chain, which is contained in $p z r$, which in turn is contained in \C{p}{q}, it follows that $v$ is a vertex visible to $p$ contained in \C{p}{q}. Furthermore, since the polygon defined by $p z$, $z r$ and the convex chain does not contain any constraints, $v$ lies in \reg{\C{p}{q}}{p}{q}. Thus, it follows from Lemma~\ref{lem:VisibleVertex} that there exists a vertex in \C{p}{q} that is visible to both $p$ and $q$, contradicting that $p q$ is an edge of the \delGraph. 
\end{proof}

\subsection{Spanning Ratio}
\label{sec:SpanningRatioDelaunay}
Let $x$ and $y$ be two distinct points on the boundary $\partial C$ of $C$. These two points split $\partial C$ into two parts. For each of these parts, there exists an isosceles triangle with base $x y$ such that the third vertex lies on that part of $\partial C$. We denote the base angles of these two triangles by $\alpha_{x, y}$ and $\alpha'_{x, y}$. We define $\alpha_C$ as follows: 
\begin{equation}
\label{eq:Diamond}
\alpha_C = \min \{ \max (\alpha_{x, y}, \alpha'_{x, y}) : x,y \in \partial C,~x \neq y \}.
\end{equation}
Note that since this function is defined on a compact set, the minimum and maximum exist and this function is well-defined. Some examples of $\alpha_C$ are the following: When $C$ is a circle, $\alpha_C = \pi/4$, when $C$ is a rectangle where $l$ and $s$ are the length of its long and short side, $\alpha_C = \tan^{-1}(s/l)$, and when $C$ is an equilateral triangle, $\alpha_C = \pi/3$. 

Given a graph $G$ and an angle $0 < \alpha < \pi/2$, we say that an edge $p q$ of $G$ satisfies the \emph{$\alpha$-diamond property}, when at least one of the two isosceles triangles with base $p q$ and base angle $\alpha$ does not contain any vertex visible to both $p$ and $q$. A graph $G$ satisfies the $\alpha$-diamond property when all of its edges satisfy this property~\cite{DJ89}. 

\begin{lemma}
  \label{lem:DiamondProperty}
  Let $C$ be any convex shape. The \delGraph satisfies the $\alpha_C$-diamond property. 
\end{lemma}
\begin{proof}  
  Let $p q$ be any edge of the \delGraph. Since $p q$ is an edge, there exists a \C{p}{q} such that \C{p}{q} does not contain any vertices that are visible to both $p$ and $q$. The vertices $p$ and $q$ split the boundary $\partial \C{p}{q}$ of \C{p}{q} into two parts and each of these parts defines an isosceles triangle with base $p q$. Let $\beta$ and $\gamma$ be the base angles of these two isosceles triangles and assume without loss of generality that $\beta \geq \gamma$ (see Figure~\ref{fig:DiamondProperty}). Let $x$ be the third vertex of the isosceles triangle having base angle $\beta$. 

  \begin{figure}[ht]
    \begin{center}
      \includegraphics{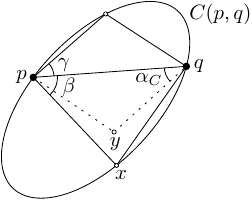}
    \end{center}
    \caption{The \delGraph satisfies the $\alpha_C$-diamond property.}
    \label{fig:DiamondProperty}
  \end{figure}

  Since $p \neq q$ and both lie on the boundary of \C{p}{q}, the pair $\{\beta, \gamma\}$ is one of the pairs considered when determining $\alpha_C$ in Equation~\ref{eq:Diamond}. Hence, since $\beta \geq \gamma$, it follows that $\alpha_C \leq \beta$. Let $y$ be the third point of the isosceles triangle having base $p q$ and base angle $\alpha_C$ that lies on the same side of $p q$ as triangle $p x q$ (see Figure~\ref{fig:DiamondProperty}). Since $\alpha_C \leq \beta$, triangle $p y q$ is contained in triangle $p x q$. By convexity of \C{p}{q}, $p x q$ is contained in \C{p}{q}. Hence, since \C{p}{q} does not contain any vertices visible to both $p$ and $q$, triangle $p y q$ does not contain any vertices visible to both $p$ and $q$ either. Hence, $p q$ satisfies the $\alpha_C$-diamond property. 
\end{proof}

For the next property, fix $O$ to be a point in the interior of $C$. Let $x$ and $y$ be two distinct points on $\partial C$, such that $x$, $y$, and $O$ are collinear. Again, $x$ and $y$ split $\partial C$ into two parts. Let $\ell_{x,y}$ and $\ell'_{x,y}$ denote the lengths of these two parts. We define $\kappa_{C,O}$ as follows: \[\kappa_{C,O} = \max \left\{ \frac{\max (\ell_{x,y}, \ell'_{x,y})}{|x y|} : x, y \in \partial C,~x \neq y, \text{ and } x, y, \text{ and } O \text{ are collinear} \right\}.\]

We note that the \delGraph does not depend on the location of $O$ inside $C$, as the presence of any edge $p q$ is defined in terms of \C{p}{q}, which does not depend on the location of $O$. Therefore, we define $\kappa_C$ as follows: \[\kappa_C = \min \{ \kappa_{C,O} : O \text{ is in the interior of } C \}.\] Throughout the remainder of this section, we assume that $O$ is picked such that $\kappa_C = \kappa_{C,O}$. We refer to this $O$ as the \emph{center} of $C$. Some examples of $\kappa_C$ are the following: When $C$ is a circle, $\kappa_C = \pi/2$ with $O$ being the center of $C$, when $C$ is a rectangle where $l$ and $s$ are the length of its long and short side, $\kappa_C = (l + s) / s = l/s + 1$ with $O$ being the center of $C$, and when $C$ is an equilateral triangle, $\kappa_C = \sqrt{3}$ with $O$ being the center of mass of $C$. 

Given a \delGraph $G$, let $p$ and $q$ be two vertices on the boundary of a face $f$ of the \delGraph, such that $p$ can see $q$ (i.e., $p q$ does not intersect any constraints) and the line segment $p q$ does not intersect the exterior of $f$. If for every such pair $p$ and $q$ on every face $f$, there exists a path in $G$ of length at most $\kappa \cdot |p q|$, then $G$ satisfies the \emph{visible-pair $\kappa$-spanner property}. We show that the \delGraph satisfies the visible-pair $\kappa_C$-spanner property. However, before we do this, we bound the length of the union of the boundary of a sequence of homothets that have their centers on a line. 

Let a set of $k + 1$ vertices $v_1, ..., v_{k+1}$ be given, such that all vertices lie on one side of the line through $v_1$ and $v_{k+1}$. For ease of exposition, assume the line through $v_1$ and $v_{k+1}$ is the $x$-axis and all vertices lie on or above this line. We consider only point sets for which there exists $C_1, ..., C_k$, a set of homothets of $C$, such that the center of each homothet lies on the $x$-axis, $C_i$ has $v_i$ and $v_{i+1}$ on its boundary, and no $C_i$ contains any vertices other than $v_i$ and $v_{i+1}$, for all $i \in \{ 1, ..., k-1 \}$. Let $\partial C^+$ be the boundary of $C$ above the $x$-axis and let $\partial (v_i, v_{i+1})$ be the part of $\partial C^+_i$ between $v_i$ and $v_{i+1}$. 

\begin{lemma}
  \label{lem:UpperHull}
  Let \C{v_1}{v_{k+1}} be the homothet of $C$ with $v_1$ and $v_{k+1}$ on its boundary and its center on the $x$-axis. It holds that \[\sum_{i=1}^k |\partial (v_i, v_{i+1})| \leq |\partial C^+(v_1, v_{k+1})|.\]
\end{lemma}
\begin{proof}
  We prove the lemma by induction on $k$, the number of homothets. If $k = 1$, $\partial (v_1, v_2)$ is $\partial C^+(v_1, v_2)$, so the lemma holds. 

  If $k > 1$, we assume that the induction hypothesis holds for all sets of at most $k-1$ homothets. Since homothet $C_i$ does not contain any vertices other than $v_i$ and $v_{i+1}$, it follows that none of the homothets are fully contained in the union of the other homothets.

  Order the homothets by increasing value of their right intersection point with the $x$-axis. Thus, $C_1$ has the smallest right intersection point and $C_k$ has the largest. Let $r$ be the right intersection point of $C_{k-1}$ and let $l$ be the left intersection point of $C_k$ and the $x$-axis (see Figure~\ref{fig:UpperHull}). Let $\partial (v_1, v_k) = \bigcup_{i=1}^{k-1} \partial (v_i, v_{i+1})$ and let $\partial (v_k, r)$ be the part of $\partial C^+_{k-1}$ between $v_k$ and $r$. Let $\partial (l, v_k)$ be the part of $\partial C^+_k$ between $l$ and $v_k$. Since $\sum_{i=1}^k |\partial (v_i, v_{i+1})| = |\partial (v_1, v_k)| + |\partial (v_k, v_{k+1})|$, to prove the lemma, we show that $|\partial (v_1, v_k)| + |\partial (v_k, v_{k+1})| \leq |\partial C^+(v_1, v_{k+1})|$. 

  \begin{figure}[ht]
    \begin{center}
      \includegraphics{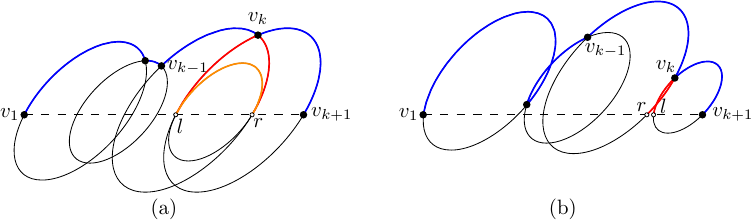}
    \end{center}
    \caption{The partial boundaries $\partial (v_1, v_k)$ and $\partial (v_k, v_{k+1})$ (blue), $\partial (l, v_k)$ and $\partial (v_k, r)$ (red), and $\partial \C{l}{r}$ (orange): (a) $l$ lies to the left of $r$, (b) $l$ lies on or to the right of $r$.}
    \label{fig:UpperHull}
  \end{figure}

  Let $c = |\partial C^+(v_1, v_{k+1})| / |v_1 v_{k+1}|$, so $|\partial C^+(v_1, v_{k+1})| = c \cdot |v_1 v_{k+1}|$. Since \newline $|\partial (v_1, v_k)| + |\partial (v_k, r)| = |\partial (v_1, v_{k-1})| + |\partial (v_{k-1}, r)|$, it follows from the induction hypothesis that $|\partial (v_1, v_k)| + |\partial (v_k, r)| = |\partial (v_1, v_{k-1})| + |\partial (v_{k-1}, r)| \leq |\partial C^+(v_1, r)| = c \cdot |v_1 r|$. Since the center of $C_k$ lies on the $x$-axis, it follows that $|\partial C_k^+| = |\partial (l, v_k)| + |\partial (v_k, v_{k+1})| = c \cdot |l v_{k+1}|$. We consider two cases: (a) $l$ lies to the left of $r$, (b) $l$ lies on or to the right of $r$.

  \textbf{Case (a):} If $l$ lies to the left of $r$, let \C{l}{r} be the homothet centered on the $x$-axis with $l$ and $r$ on its boundary (see Figure~\ref{fig:UpperHull}a). Hence, it follows that $|\partial C^+(l, r)| = c \cdot |l r|$. Since \C{l}{r} has $l$ and on its left boundary, it is contained in $C_k$, and since it has $r$ on its right boundary, it is contained in $C_{k-1}$. Hence, \C{l}{r} is contained in the intersection of $C_{k-1}$ and $C_k$. Since the length of the boundary of this intersection above the $x$-axis is $|\partial (l, v_k)| + |\partial (v_k, r)|$ and \C{l}{r} is convex, it follows that $|\partial C^+(l, r)| \leq |\partial (l, v_k)| + |\partial (v_k, r)|$. Hence, we have that 
  \begin{eqnarray*}
    \sum_{i=1}^k |\partial (v_i, v_{i+1})| &=& |\partial (v_1, v_k)| + |\partial (v_k, v_{k+1})| \\
    &=& |\partial (v_1, v_k)| + |\partial (v_k, r)| - |\partial (v_k, r)| + |\partial C^+_k| - |\partial (l, v_k)| \\
    &\leq& c \cdot |v_1 r| - |\partial (v_k, r)| + c \cdot |l v_{k+1}| - |\partial (l, v_k)| \\
    &=& c \cdot |v_1 v_{k+1}| + c \cdot |l r| - |\partial (l, v_k)| - |\partial (v_k, r)| \\
    &=& |\partial C^+(v_1, v_{k+1})| + |\partial C^+(l, r)| - |\partial (l, v_k)| - |\partial (v_k, r)| \\
    &\leq& |\partial C^+(v_1, v_{k+1})|.
  \end{eqnarray*}

  \textbf{Case (b):} If $l$ lies on or to the right of $r$ (see Figure~\ref{fig:UpperHull}b), we have that
  \begin{eqnarray*}
    \sum_{i=1}^k |\partial (v_i, v_{i+1})| &=& |\partial (v_1, v_k)| + |\partial (v_k, v_{k+1})| \\
    &\leq& |\partial (v_1, v_k)| + |\partial (v_k, r)| + |\partial (l, v_k)| + |\partial (v_k, v_{k+1})| \\
    &\leq& c \cdot |v_1 r| + c \cdot |l v_{k+1}| \\
    &\leq& c \cdot |v_1 v_{k+1}| \\
    &=& |\partial C^+(v_1, v_{k+1})|,
  \end{eqnarray*}
  completing the proof. 
\end{proof}

\begin{lemma}
  \label{lem:VisiblePairSpanner}
  The \delGraph satisfies the visible-pair $\kappa_C$-spanner property. 
\end{lemma}
\begin{proof}
  Let $p$ and $q$ be two vertices on the boundary of a face $f$ of the \delGraph, such that $p$ can see $q$ and the line segment $p q$ does not intersect the exterior of $f$. Assume without loss of generality that $p q$ lies on the $x$-axis. Let \C{p}{q} be the homothet of $C$ with $p$ and $q$ on its boundary and its center on $p q$. We aim to show that there exists a path between $p$ and $q$ of length at most $\kappa_C \cdot |p q|$. Since by definition $\kappa_C$ is at least $|\partial C^+(p, q)| / |p q|$, showing that there exists a path between $p$ and $q$ of length at most $|\partial C^+(p, q)|$ completes the proof. If $p q$ is an edge of the \delGraph, this follows from the triangle inequality, so assume this is not the case. 

  We grow a homothet $C'$ with its center on $p q$ by moving its center from $p$ to $q$, while maintaining that $p$ lies on the boundary of $C'$ (see Figure~\ref{fig:ConstructingPath}a). Let $v_1$ be the first vertex hit by $C'$ that is visible to $p$ and lies in \reg{\C{p}{q}}{p}{q}. We assume without loss of generality that $v_1$ lies above $p q$. Since $v_1$ is the first vertex satisfying these conditions, $p v_1$ is either an edge or a constraint: Since $v_1$ is the first visible vertex we hit in \reg{\C{p}{q}}{p}{q}, we have that $\reg{\C{p}{q}}{p}{q} \cap C'$ contains no vertices visible to $p$. Hence, there is no vertex visible to both $p$ and $v_1$. Therefore, Lemma~\ref{lem:VisibleVertex} implies that $\reg{\C{p}{q}}{p}{q} \cap C'$ does not contain any vertices visible to $v_1$. Hence, if $p v_1$ is not a constraint, the region that is visible to both $p$ and $v_1$ does not contain any vertices and $p v_1$ is an edge of the \delGraph. 

  \begin{figure}[ht]
    \begin{center}
      \includegraphics{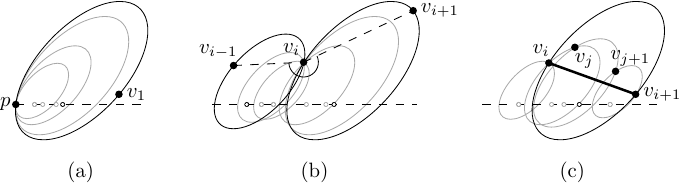}
    \end{center}
    \caption{Constructing a path from $p$ to $q$: (a) growing $C'$ from $p$, (b) growing $C'$ while maintaining that $v_i$ lies on its boundary, (c) refining when $v_i v_{i+1}$ is a constraint.}
    \label{fig:ConstructingPath}
  \end{figure}

  We continue constructing a sequence of vertices $p, v_1, v_2, ..., v_k, q$ until we hit $q$ by moving the center of $C'$ along $p q$ towards $q$ and each time we hit a vertex $v_i$, we require that it lies on the boundary of $C'$ until we hit the next vertex $v_{i+1}$ that is visible to $v_i$ and $v_i$ is not the endpoint of a constraint that lies in the counterclockwise angle $\angle v_{i-1} v_i v_{i+1}$ (see Figure~\ref{fig:ConstructingPath}b). Since $v_{i+1}$ is the first vertex satisfying these conditions starting from $v_i$, we know that $v_i v_{i+1}$ is either an edge or a constraint by the same argument used above to show that $p v_1$ is an edge or a constraint. Since $p q$ is visible and does not intersect the exterior of the face $f$, this in turn implies that these vertices all lie above $p q$. 

  Unfortunately, we cannot assume that there exists an edge between every pair of consecutive vertices: If $v_i v_{i+1}$ is a constraint, there can be vertices visible to both $v_i$ and $v_{i+1}$ on the opposite side of the constraint. For pairs of \mbox{vertices $v_i, v_{i+1}$} that do not form an edge, we refine the construction of the sequence \mbox{between them:} We start with $C'$ such that it does not cross $v_i v_{i+1}$ and $v_i$ lies on its boundary. We construct a sequence of vertices from $v_i$ to $v_{i+1}$ by moving the center of $C'$ along $p q$ towards $q$, maintaining that $v_i$ lies on its boundary (see Figure~\ref{fig:ConstructingPath}c). For the first vertex we hit, we require that it is visible to $v_i$ and lies in \reg{C'}{v_i}{v_{i+1}}. 

  We continue moving the center of $C'$ along $p q$ towards $q$, but we now maintain that $v_i$ lies on the boundary of $C'$. Each time we hit a vertex $v_j$, we require that it lies on the boundary of $C'$ until we hit the next vertex $v_{j+1}$ that is visible to $v_j$ and $v_j$ is not the endpoint of a constraint that lies in the counterclockwise angle $\angle v_{j-1} v_j v_{j+1}$. In other words, we construct a more fine-grained sequence when consecutive vertices define a constraint and there is no edge between them. Note that we may need to repeat this process a number of times, since there need not be edges between the vertices of the finer grained sequence either. However, since the point set is finite, this process terminates. 

  This way, we end up with a path $p, v_1, v_2, ..., v_l, q$ from $p$ to $q$ that lies above $p q$. Furthermore, since $C$ is convex, we can upper bound the length of each edge $v_i v_{i+1}$ by the part of $\partial \C{v_i}{v_{i+1}}$, the homothet with $v_i$ and $v_{i+1}$ on its boundary and its center on $p q$, that does not intersect $p q$. Hence, the total length of the path is upper bounded by the length of the union of the boundaries of these homothets above $p q$. By construction, none of the homothets corresponding to consecutive vertices along the path contain any of the other vertices along the path. Hence, we can apply Lemma~\ref{lem:UpperHull} and it follows that the total length of the path is at most $|\partial C^+(p, q)|$, completing the proof. 
\end{proof}

A path between two vertices $p$ and $q$ is called \emph{one-sided} if all vertices along this path lie above the line through $p q$ or all vertices lie below the line through $p q$. Since the path constructed in Lemma~\ref{lem:VisiblePairSpanner} is one-sided, we get the following corollary. 

\begin{corollary}
\label{coro:OneSided}
  The spanning ratio of any one-sided path in $G$ is at most $\kappa_C$.
\end{corollary}

We are now ready to prove that the \delGraph is a spanner. Das and Joseph~\cite{DJ89} showed that any plane graph that satisfies the $\alpha$-diamond property and the good polygon property (similar to the visible-pair $\kappa$-spanner property) is a spanner. Subsequently, Bose~\etal~\cite{BLS07} improved slightly on the spanning ratio. They showed that a geometric (constrained) graph $G$ is a spanner of the visibility graph when it satisfies the following properties: 
\begin{enumerate}
  \item $G$ is plane. 
  \item $G$ satisfies the $\alpha$-diamond property.
  \item The spanning ratio of any one-sided path in $G$ is at most $\kappa$.
  \item $G$ satisfies the visible-pair $\kappa'$-spanner property. 
\end{enumerate}
In particular, $G$ is a $t$-spanner for \[t = 2 \kappa \kappa' \cdot \max \left( \frac{3}{\sin(\alpha/2)}, \kappa \right).\]

It follows from Lemmas~\ref{lem:PlaneDelaunay}, \ref{lem:DiamondProperty}, and~\ref{lem:VisiblePairSpanner}  and Corollary~\ref{coro:OneSided} that the \delGraph satisfies these four properties. Moreover, even though in general the \delGraph is not a triangulation, if for a specific \circle it is, it satisfies the visible-pair 1-spanner property: Since every face consists of three vertices that are pairwise connected by an edge, the shortest path between two vertices $p$ and $q$ on this face has length $1 \cdot |p q|$. Therefore, we obtain the following theorem: 

\begin{theorem}
  \label{theo:GeneralizedDelaunaySpanningRatio}
  The \delGraph $G$ is a $t$-spanner of $\Vis(P,S)$ for 
  \begin{equation*}
    t = 
    \begin{cases}
      2 \kappa_C \cdot \max \left( \frac{3}{\sin(\alpha_C/2)}, \kappa_C \right), & \text{ \emph{if $G$ is a triangulation}} \\
      2 \kappa_C^2 \cdot \max \left( \frac{3}{\sin(\alpha_C/2)}, \kappa_C \right), & \text{ \emph{otherwise.}}
    \end{cases}
  \end{equation*}
\end{theorem}

Though this theorem holds for all convex shapes, the bound it provides is rather loose when we look at a specific shape. For example, for the constrained Delaunay graph that uses an equilateral triangle, the above theorem implies an upper bound of $2 \cdot \sqrt{3}^2 \cdot 3/\sin(\pi/6) = 36$, which is far greater than the tight bound of 2~\cite{BFRV12}. For circles, the best known upper bound is 1.998~\cite{X13}, while Theorem~\ref{theo:GeneralizedDelaunaySpanningRatio} implies a ratio of $\pi \cdot 3 \cdot \sqrt{4 + 2\sqrt{2}} \approx 24.63$, since these graphs are triangulations. For squares, Theorem~\ref{theo:GeneralizedDelaunaySpanningRatio} implies a ratio of $24 \cdot \sqrt{4 + 2\sqrt{2}}  \approx 62.72$ compared to the tight ratio of 2.61 in the unconstrained setting~\cite{BGHP12} and for rectangles, we get an upper bound of $2 \cdot (l/s + 1)^2 \cdot 3 \sqrt{l^2 + s^2} / s > 2 \cdot (l/s + 1)^3$.

\section[The Constrained Empty-Rectangle Delaunay Graph]{The Constrained Empty-Rectangle Delaunay Graph}
\label{sec:rectangles}
In this section, we look at the case where the empty \circle is an arbitrary rectangle and reduce the dependency of the spanning ratio on the aspect ratio from cubic, which is implied by Theorem~\ref{theo:GeneralizedDelaunaySpanningRatio}, to linear. To this end, we first take a closer look at empty visibility regions in the convex shapes. Next, we take a closer look at Lemma~\ref{lem:UpperHull} for the case of rectangles, as it will be convenient to explicitly argue about the lengths of the edges of the spanning path in terms of the sides of the rectangle. Following this, we use these two lemmas to bound the length of a path when \C{p}{q} is known to be partially empty. Finally, we use this latter lemma to arrive at the desired result. 

We assume without loss of generality that the rectangle is axis-aligned. We do not, however, assume anything about the ratio between the height and width of the rectangle. We first prove an auxiliary lemma that will be used to show that certain regions of the rectangles are empty. In the interest of possible future use, we prove this lemma for an arbitrary convex shape and apply it only to the case of rectangles in our proof. 

\begin{lemma}
  \label{lem:VisibleBelowPQ}
  Let $p$ and $q$ be two vertices that can see each other and let \C{p}{q} be any convex shape with $p$ and $q$ on its boundary. Let $H_1$ and $H_2$ be the intersection of \C{p}{q} with the two half-planes defined by the line through $pq$, respectively. If there exists a  point $x$ in $H_2$ can see a vertix in $H_1$, then $H_1$ contains a vertex visible to $p$ and $q$. 
\end{lemma}
\begin{proof}
  We assume without loss of generality that $p q$ is not vertical. We also assume without loss of generality that $H_1$ is the intersection of \C{p}{q} with the half-plane below the line through $pq$ and $H_2$ is the intersection of \C{p}{q} with the half-plane above the line through $pq$. We prove the lemma by contradiction, so assume that there exists a vertex $y$ in \C{p}{q} below $p q$ that is visible to $x$, but not to $p$ and $q$. Since \C{p}{q} is a convex shape and $x$ and $y$ lie on opposite sides of $p q$, the visibility edge $x y$ intersects $p q$. Let $z$ be this intersection (see Figure~\ref{fig:VisibleBelowPQ}). 

  \begin{figure}[ht]
    \begin{center}
      \includegraphics{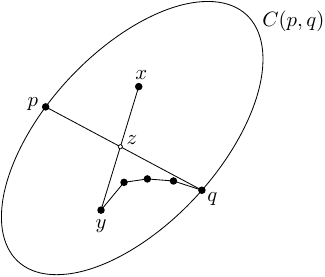}
    \end{center}
    \caption{If $x$ can see a vertex below $p q$, then so can $q$.}
    \label{fig:VisibleBelowPQ}
  \end{figure}

  Hence, $z y$ and $z q$ are visibility edges. Since $z$ is not a vertex, it is not the endpoint of any constraints intersecting the interior of triangle $y z q$. It follows from Lemma~\ref{lem:ConvexChain} that there exists a convex chain of visibility edges between $y$ and $q$ and this chain is contained in $y z q$. However, this implies that $w$, the neighbor of $q$ along this chain, is visible to $q$ and lies in \C{p}{q} below $p q$. Next, we apply Lemma~\ref{lem:ConvexChain} on triangle $p q w$ and find that the neighbor of $p$ along the chain from $p$ to $w$ is visible to both $p$ and $q$ and lies in \C{p}{q} below $p q$, contradicting that this region does not contain any vertices visible to $p$ and $q$. 
\end{proof}

Next, we revisit Lemma~\ref{lem:UpperHull}, since in the remainder of the spanning proof, it is convenient to explicitly argue about the lengths of the edges of the spanning path in terms of the sides of the rectangle. 

We first introduce some notation for the following lemma. Let $p$ and $q$ be two vertices of the \delGraph that can see each other. Let $R$ be a rectangle with $p$ and $q$ on its West and East boundary and let $a$, $b$, and $r$ be the Northwest, Northeast, and Southwest corner of $R$. Let $m_1, ..., m_{k-1}$ be any $k-1$ points on $p q$ in the order they are visited when walking from $p$ to $q$ (see Figure~\ref{fig:SummingRectangles}). Let $m_0 = p$ and $m_k = q$. Consider the homothets $S_i$ of $R$ with $m_i$ and $m_{i+1}$ on their respective boundaries, for $0 \leq i < k$, such that $|p a| / |r a| = |m_i a_i| / |r_i a_i|$, where $a_i$, $b_i$, $r_i$ are the Northwest, Northeast, and Southwest corner of $S_i$. 

\begin{figure}[ht]
  \begin{center}
    \includegraphics{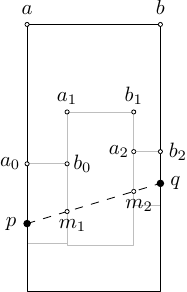}
  \end{center}
  \caption{The total length of the sides of the rectangles $S_i$ equals that of \C{p}{q}.}
  \label{fig:SummingRectangles}
\end{figure}

\begin{lemma}
  \label{lem:SummingRectangles}
   We have \[\sum_{i=0}^{k-1} \big( |m_i a_i| + |a_i b_i| + |b_i m_{i+1}| \big) = |p a| + |a b| + |b q|.\]
\end{lemma}
\begin{proof}
  Let $c = (|p a| + |a b| + |b q|) / |p q|$. We first show that for every $S_i$ we have that $(|m_i a_i| + |a_i b_i| + |b_i m_{i+1}|) / |m_i m_{i+1}| = c$, for $0 \leq i < k$. Since $S_i$ is a homothet of $R$ and the slopes of $m_i m_{i+1}$ and $p q$ are equal, we have that $|a b| / |p q| = |a_i b_i| / |m_i m_{i+1}|$. Furthermore, by construction $|p a| / |r a| = |m_i a_i| / |r_i a_i|$, and since the slopes of $m_i m_{i+1}$ and $p q$ are equal, we also have that $|b q| / |r a| = |b_i m_{i+1}| / |r_i a_i|$. Finally, since $S_i$ is a homothet of $R$, we have that $|p q| / |r a| = |m_i m_{i+1}| / |r_i a_i|$, which gives $|p a| / |p q| = |m_i a_i| / |m_i m_{i+1}|$ and $|b q| / |p q| = |b_i m_{i+1}| / |m_i m_{i+1}|$. 
  
Hence, since $(|m_i a_i| + |a_i b_i| + |b_i m_{i+1}|) / |m_i m_{i+1}| = c$, for $0 \leq i < k$, we get 
  \begin{align*}
    \sum_{i=0}^{k-1} \big( |m_i a_i| + |a_i b_i| + |b_i m_{i+1}| \big) & = \sum_{i=0}^{k-1} \big( c \cdot |m_i m_{i+1}| \big) \\
    & = c \cdot |p q| \\
    & = |p a| + |a b| + |b q|,
  \end{align*}
  proving the lemma. 
\end{proof}

Before we prove the bound on the spanning ratio of the \delGraph, we first bound the length of the spanning path between vertices $p$ and $q$ for the case where the rectangle \C{p}{q} is partially empty. We call a rectangle \C{p}{q} \emph{half-empty} when \C{p}{q} contains no vertices in \reg{\C{p}{q}}{p}{q} below $p q$ that are visible to $p$ and no vertices in \reg{\C{p}{q}}{q}{p} below $p q$ that are visible to $q$. We denote the $x$- and $y$-coordinate of a point $p$ by $p_x$ and $p_y$. 

\begin{lemma}
  \label{lem:HalfEmptyRectangle}
  Let $p$ and $q$ be two vertices that can see each other. Let \C{p}{q} be a rectangle with $p$ and $q$ on its boundary, such that it is half-empty. Let $a$ and $b$ be the corners of \C{p}{q} on the non-half-empty side. The \delGraph contains a path between $p$ and $q$ of length at most $|p a| + |a b| + |b q|$.
\end{lemma}
\begin{proof}
  We prove the lemma by induction on the rank of \C{x}{y} when ordered by size, for any two visible vertices $x$ and $y$, such that \C{x}{y} is half-empty. We assume without loss of generality that $p$ lies on the West boundary, $q$ lies on the East boundary and that \C{p}{q} is half-empty below $p q$. This implies that $a$ and $b$ are the Northwest and Northeast corner of \C{p}{q}, respectively. We also assume without loss of generality that the slope of $p q$ is non-negative, i.e., $p_x < q_x$ and $p_y \leq q_y$ (see Figure~\ref{fig:RectanglePath}). Note that this can be achieved by swapping $p$ and $q$, if needed. 
  
  \begin{figure}[ht]
    \vspace{-1em}
    \begin{center}
      \includegraphics{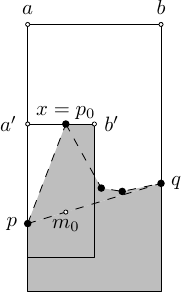}
    \end{center}
    \vspace{-0.05em}
    \caption{An inductive path from $p$ to $q$.}
    \label{fig:RectanglePath}
  \end{figure}
  
  We note that the case where $p$ lies on the West boundary, $q$ lies on the North boundary and \C{p}{q} is half-empty below $p q$ can be viewed as a special case of the one above: We shrink \C{p}{q} until one of $p$ or $q$ lies in a corner. This point can now be viewed as being on both sides defining the corner and hence $p$ and $q$ are on opposite sides: If $p$ lies in the Southwest corner, we treat it as lying on the South boundary when $q$ lies on the North boundary. If $q$ lies in the Northeast corner, we treat it as lying on the East boundary when $p$ lies on the West boundary. Analogous statements hold for the case where $p$ lies on the West boundary, $q$ lies on the North boundary and \C{p}{q} is half-empty above $p q$. 
  
  Let $r$ be the Southwest corner of \C{p}{q}. Let $R$ be a homothet of \C{p}{q} that is contained in \C{p}{q} and whose West boundary is intersected by $p q$. Let $a'$, $b'$, $r'$ be the Northwest, Northeast, and Southwest corner of $R$ and let $m$ be the intersection of $a' r'$ and $p q$. We call homothet $R$ \emph{similar} to \C{p}{q} if and only if $|p a| / |r a| = |m a'| / |r' a'|$. 

  \textbf{Base case:} If \C{p}{q} is a half-empty rectangle of smallest area, then \C{p}{q} does not contain any vertices visible to both $p$ and $q$: Assume this is not the case and grow a rectangle $R$ similar to \C{p}{q} from $p$ to $q$. Let $x$ be the first vertex hit by $R$ that is visible to $p$ and lies in \reg{\C{p}{q}}{p}{q}. Note that this implies that $R$ is contained in \C{p}{q}. Therefore, $R$ is smaller than \C{p}{q}. Furthermore, $R$ is half-empty, since by Lemma~\ref{lem:VisibleBelowPQ}, the part below the line through $p$ and $q$ does not contain any vertices visible to $p$ or $x$ in \reg{\C{p}{q}}{p}{q}, and the part between the line through $p$ and $x$ and the line through $p$ and $q$ does not contain any vertices visible to $p$ or $x$ since $x$ is the first visible vertex hit while growing $R$. However, this contradicts that \C{p}{q} is the smallest half-empty rectangle. 

  Hence, \C{p}{q} does not contain any vertices visible to both $p$ and $q$, which implies that $p q$ is an edge of the \delGraph. Therefore the length of the shortest path from $p$ to $q$ is at most $|p q| \leq |p a| + |a b| + |b q|$. 

  \textbf{Induction step:} We assume that for all half-empty rectangles \C{x}{y} smaller than \C{p}{q} the lemma holds. If $p q$ is an edge of the \delGraph, the length of the shortest path from $p$ to $q$ is at most $|p q| \leq |p a| + |a b| + |b q|$. 

  If $p q$ is not an edge of the \delGraph, there exists a vertex in \C{p}{q} that is visible from both $p$ and $q$. We grow a rectangle $R$ similar to \C{p}{q} from $p$ to $q$. Let $x$ be the first vertex hit by $R$ that is visible to $p$ and lies in \reg{\C{p}{q}}{p}{q} and let $a'$ and $b'$ be the Northwest and Northeast corner of $R$ (see Figure~\ref{fig:RectanglePath}). Note that this implies that $R$ is contained in \C{p}{q}. We also note that $p x$ is not necessarily an edge in the \delGraph, since if it is a constraint, there can be vertices visible to both $p$ and $x$ above $p x$ inside $R$. However, since $R$ is half-empty and smaller than \C{p}{q}, we can apply induction on it and we obtain that the path from $p$ to $x$ has length at most $|p a'| + |a' b'| + |b' x|$ when $x$ lies on the East boundary of $R$, and that the path from $p$ to $x$ has length at most $|p a'| + |a' x|$ when $x$ lies on the North boundary of $R$.

  \textbf{Bounding the path length between visible vertices:} Let $m_0$ be the projection of $x$ along the vertical axis onto $p q$. Since $m_0$ is contained in $R$, $x$ can see $m_0$. Since $x m_0$ and $m_0 q$ are visibility edges and $m_0$ is not the endpoint of a constraint intersecting the interior of triangle $x m_0 q$, we can apply Lemma~\ref{lem:ConvexChain} and obtain a convex chain $x = p_0, p_1, ..., p_k = q$ of visibility edges (see Figure~\ref{fig:RectanglePath}). For each of these visibility edges $p_i p_{i+1}$, there is a homothet $R_i$ of \C{p}{q} that falls in one of the following three types (see Figure~\ref{fig:RectanglesAlongChain}): (i) $p_i$ lies on the North boundary and $p_{i+1}$ lies in the Southeast corner, (ii) $p_i$ lies on the West boundary and $p_{i+1}$ lies on the East boundary and the slope of $p_i p_{i+1}$ is negative, (iii) $p_i$ lies on the West boundary and $p_{i+1}$ lies on the East boundary and the slope of $p_i p_{i+1}$ is not negative. Note that the case where $p_i$ lies on the South boundary and $p_{i+1}$ lies on the North boundary cannot occur, since the slope of any $p_i p_{i+1}$ is at most that of $p q$. Also note that the case where $p_i$ lies on the South boundary and $p_{i+1}$ lies on the East boundary cannot occur, since we can shrink the rectangle until $p$ lies in the Southwest corner, resulting in a Type~(iii) rectangle. Let $a_i$ and $b_i$ be the Northwest and Northeast corner of $R_i$. We note that by convexity, these three types occur in the order Type~(i), Type~(ii), and Type~(iii) along the convex chain from $x$ to $q$. 

  \begin{figure}[ht]
    \begin{center}
      \includegraphics{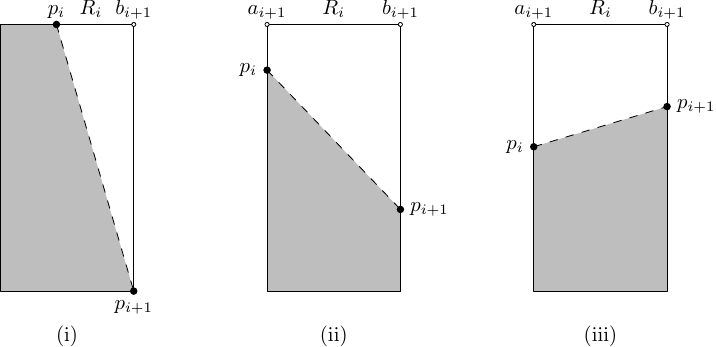}
    \end{center}
    \caption{The three types of rectangles along the convex chain.}
    \label{fig:RectanglesAlongChain}
  \end{figure}

  Let $m_i$ be the projection of $p_i$ along the vertical axis onto $p q$, let $C_i$ be the homothet of \C{p}{q} with $m_i$ and $m_{i+1}$ on its boundary that is similar to \C{p}{q}, and let $a_i'$ and $b_i'$ be the Northwest and Northeast corner of $C_i$. Using these $C_i$, we  shift Type~(ii) and Type~(iii) rectangles down as far as possible: We shift $R_i$ down until either $p_i$ or $p_{i+1}$ lies in one of the North corners or the South boundary corresponds to the South boundary of $C_i$. In the latter case, $R_i$ and $C_i$ are the same rectangle. 

  Since all rectangles $R_i$ are smaller than \C{p}{q}, we can apply induction, provided that we can show that $R_i$ is half-empty. For Type~(i) visibility edges, the part of the rectangle that lies below the line through $p_i$ and $p_{i+1}$ is contained in $R$, which does not contain any visible vertices, and the region of \reg{\C{p}{q}}{p}{q} below the convex chain, which is empty. For Type~(ii) and Type~(iii) visibility edges, the part of the rectangle that lies below the line through $p_i$ and $p_{i+1}$ is contained in the region of \reg{\C{p}{q}}{p}{q} below the convex chain, which is empty, and the region of \C{p}{q} below the line through $p$ and $q$, which does not contain any visible vertices by Lemma~\ref{lem:VisibleBelowPQ}. Hence, all $R_i$ are half-empty and we obtain an inductive path of length at most: (i) $|p_i b_i| + |b_i p_{i+1}|$ for Type~(i) rectangles, (ii) $|p_i a_i| + |a_i b_i| + |b_i p_{i+1}|$ for Type~(ii) rectangles, (iii) $|p_i a_i| + |a_i b_i| + |b_i p_{i+1}|$ for Type~(iii) rectangles. 

  \textbf{Bounding the total path length:} To bound the total path length, we perform a case distinction on the location of $x$ on $R$ and whether the convex path from $x$ to $q$ goes down (see Figure~\ref{fig:RectangleFourCases}): (a) $x$ lies on the East boundary of $R$ and the convex path does not go down, (b) $x$ lies on the East boundary of $R$ and the convex path goes down, (c) $x$ lies on the North boundary of $R$ and the convex path does not go down, (d) $x$ lies on the North boundary of $R$ and the convex path goes down. 

  \begin{figure}[ht]
    \begin{center}
      \includegraphics{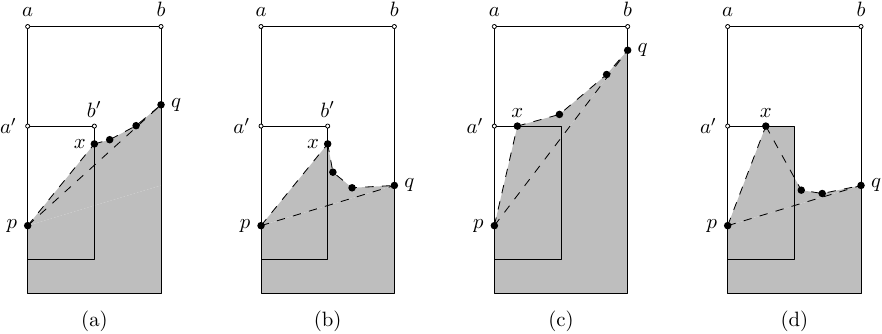}
    \end{center}
    \caption{The four cases depending on location of $x$ on $R$ and whether the convex path from $x$ to $q$ goes down.}
    \label{fig:RectangleFourCases}
  \end{figure}

  \textbf{Case (a):} The vertex $x$ lies on the East boundary of $R$ and the convex path does not go down (see Figure~\ref{fig:RectangleFourCases}a). Recall that the length of the path from $p$ to $x$ is at most $|p a'| + |a' b'| + |b' x|$, which is at most $|p a'| + |a' b'| + |b' m_0|$. Since the convex chain does not go down, it cannot contain any Type~(i) or Type~(ii) visibility edges. Furthermore, since $x$ lies on the East boundary of $R$, $R$ and all $C_i$ are disjoint. Thus, Lemma~\ref{lem:SummingRectangles} implies that the boundaries above $p q$ of $R$ and all $C_i$ sum up to $|p a| + |a b| + |b q|$. Hence, if we can show that, for all $R_i$, $|p_i a_i| + |a_i b_i| + |b_i p_{i+1}| \leq |m_i a_i'| + |a_i' b_i'| + |b_i' m_{i+1}|$, the proof of this case is complete.

  By convexity, the slope of $p_i p_{i+1}$ is at most that of $p q$ and $m_i m_{i+1}$. Hence, when $p_{i+1}$ lies in the Northeast corner of $R_i$, we have $p_{i+1} = b_i$ and $|p_i a_i| + |a_i p_{i+1}| \leq |m_i a_i'| + |a_i' b_i'| + |b_i' m_{i+1}|$. If $p_{i+1}$ does not lie in the Northeast corner, $R_i = C_i$. Hence, since $p_i$ and $p_{i+1}$ lie above $p q$, we have that $|p_i a_i| + |a_i b_i| + |b_i p_{i+1}| \leq |m_i a_i'| + |a_i' b_i'| + |b_i' m_{i+1}|$. 

  \textbf{Case (b):} The vertex $x$ lies on the East boundary of $R$ and the convex path goes down (see Figure~\ref{fig:RectangleFourCases}b). Recall that the length of the path from $p$ to $x$ is at most $|p a'| + |a' b'| + |b' x|$. Let $p_j$ be the lowest vertex along the convex chain. Since $p_j$ lies above $p q$ and $p q$ has non-negative slope, the descent of the convex path is at most $|x m_0|$. Hence, when we charge this to $R$, we used $|p a'| + |a' b'| + |b' m_0|$ of its boundary (see Figure~\ref{fig:RectanglePathEast}). 

\begin{figure}[ht]
  \begin{minipage}[t]{0.5\linewidth}
    \begin{center}
      \includegraphics{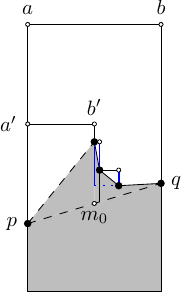}
    \end{center}
    \caption{Going down along the convex chain (blue) is charged to $R$ (orange).}
    \label{fig:RectanglePathEast}
  \end{minipage}
  \hspace{0.05\linewidth}
  \begin{minipage}[t]{0.42\linewidth}
    \begin{center}
      \includegraphics{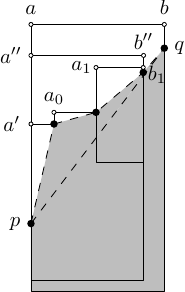}
    \end{center}
    \caption{Charging the path from $p$ to $p_j$ to \C{p}{p_j}.}
    \label{fig:RectanglePathNorthUp}
  \end{minipage}
\end{figure}

  Like in the Case~(a), since $x$ lies on the East boundary of $R$, $R$ and all $C_i$ are disjoint. Thus, Lemma~\ref{lem:SummingRectangles} implies that the boundaries above $p q$ of $R$ and all $C_i$ sum up to $|p a| + |a b| + |b q|$. Hence, if we can show that, for all $R_i$, the inductive path length is at most $|m_i a_i'| + |a_i' b_i'| + |b_i' m_{i+1}|$, the proof of this case is complete. 

  For Type~(i) visibility edges, we have already charged $|b_i p_{i+1}|$ to $R$, so it remains to show that $|p_i b_i| \leq |m_i a_i'| + |a_i' b_i'| + |b_i' m_{i+1}|$. This follows, since $m_i$ and $m_{i+1}$ are the vertical projections of $p_i$ and $p_{i+1}$, which implies that $|p_i b_i| = |a_i' b_i'|$. 

  For Type~(ii) visibility edges, we already charged $|b_i p_{i+1}| - |p_i a_i|$ to $R$, so we can consider $p_i p_{i+1}$ to be horizontal and it remains to charge the remaining $2 \cdot |p_i a_i| + |a_i b_i|$. If $p_i$ lies in the Northwest corner of $R_i$, it follows that $|p_i a_i| = 0$ and we have that $|p_i b_i| = |a_i' b_i'| \leq |m_i a_i'| + |a_i' b_i'| + |b_i' m_{i+1}|$. If $p_i$ does not lie in the Northwest corner, $R_i$ is the same as $C_i$. Hence, since we can consider $p_i p_{i+1}$ to be horizontal and $p_i$ and $p_{i+1}$ lie above $p q$, it follows that $2 \cdot |p_i a_i| + |a_i b_i| \leq |m_i a_i'| + |a_i' b_i'| + |b_i' m_{i+1}|$. 

  Finally, Type~(iii) visibility edges are charged as in Case~(a), hence we have that $|p_i a_i| + |a_i b_i| + |b_i p_{i+1}| \leq |m_i a_i'| + |a_i' b_i'| + |b_i' m_{i+1}|$, completing the proof of this case. 

  \textbf{Case (c):} Vertex $x$ lies on the North boundary of $R$ and the convex path does not go down (see Figure~\ref{fig:RectangleFourCases}c). Recall that the length of the path from $p$ to $x$ is at most $|p a'| + |a' x|$. Since the convex chain does not go down, it cannot contain any Type~(i) or Type~(ii) visibility edges. Let $p_j$ be the first vertex along the chain, such that $R_{j-1}$ is the same as $C_{j-1}$. Since $q$ lies on the East boundary of \C{p}{q}, this condition is satisfied for the last visibility edge along the convex chain, hence $p_j$ exists. 

  Let \C{p}{p_j} be the homothet of \C{p}{q} that has $p$ and $p_j$ on its boundary and is similar \C{p}{q}. Let $a''$ and $b''$ be the Northwest and Northeast corners of \C{p}{p_j} (see Figure~\ref{fig:RectanglePathNorthUp}). Since $p_j$ is first vertex along the convex chain that does not lie in the Northeast corner of $R_{j-1}$, we have that along the path from $p$ to $p_j$ the projections of $a' x$, all $a_i p_{i+1}$, and $a_{j-1} b_{j-1}$ onto $a'' b''$ are disjoint and the projections of $p a'$, all $p_i a_i$, and $p_{j-1} a_{j-1}$ onto $p a''$ are disjoint. Hence, their lengths sum up to at most $|p a''| + |a'' b''|$. Finally, since $|b_{j-1} p_j| \leq |b'' p_j|$, the total length of the path from $p$ to $p_j$ is at most $|p a''| + |a'' b''| + |b'' p_j|$, which is at most $|p a''| + |a'' b''| + |b'' m_j|$. 

  All Type~(iii) visibility edges following $p_j$ are charged as in Case~(a), hence we have that $|p_i a_i| + |a_i b_i| + |b_i p_{i+1}| \leq |m_i a_i'| + |a_i' b_i'| + |b_i' m_{i+1}|$. We now apply Lemma~\ref{lem:SummingRectangles} to \C{p}{p_j} and all $C_i$ following $p_j$ and obtain that the total length of the path from $p$ to $q$ is at most $|p a| + |a b| + |b q|$. 

  \textbf{Case (d):} Vertex $x$ lies on the North boundary of $R$ and the convex path goes down (see Figure~\ref{fig:RectangleFourCases}d). Recall that the length of the path from $p$ to $x$ is at most $|p a'| + |a' x|$ and that $p_1$ is the neighbor of $x$ along the convex chain. Let \C{p}{p_1} be the homothet of \C{p}{q} that has $p$ and $p_1$ on its boundary and is similar to \C{p}{q}. Let $a''$ and $b''$ be the Northwest and Northeast corners of \C{p}{p_1}. Since $p_1$ lies to the right of $R$ and lower than $x$, it lies on the East boundary of \C{p}{p_1}. We first show that the length of the path from $p$ to $p_1$ is at most $|p a''| + |a'' b''| + |b'' p_1|$. 

  If $x p_1$ is a Type~(i) visibility edge, the length of the path from $x$ to $p_1$ is at most $|x b_0| + |b_0 p_1|$. Hence we have a path from $p$ to $p_1$ of length at most $|p a'| + |a' x| + |x b_0| + |b_0 p_1| = |p a'| + |a'' b''| + |b_0 p_1|$. Since $|p a'| \leq |p a''|$ and $|b_0 p_1| \leq |b'' p_1|$, this implies that the path has length at most $|p a''| + |a'' b''| + |b'' p_1|$. If $x p_1$ is a Type~(ii) visibility edge and $x$ lies in the Northwest corner an analogous argument shows that the path from $p$ to $p_1$ is at most $|p a''| + |a'' b''| + |b'' p_1|$. If $x p_1$ is a Type~(ii) visibility edge and $R_0 = C_0$, we have that the projections of $a' x$ and $a_0 b_0$ onto $a'' b''$ are disjoint and the projections of $p a'$ and $x a_0$ onto $p a''$ are disjoint. Hence, their total lengths sum up to at most $|p a''| + |a'' b''|$. Finally, since $|b_0 p_1| \leq |b'' p_1|$, the total length of the path from $p$ to $p_1$ is at most $|p a''| + |a'' b''| + |b'' p_1|$. 

  Next, we observe, like in Case~(b), that starting from $p_1$ the convex path cannot go down more than $|p_1 m_1|$. Hence, when we charge this to \C{p}{p_1}, we used $|p a''| + |a'' b''| + |b'' m_1|$ of its boundary. Finally, we use arguments analogous to the ones in Case~(b) to show that each inductive path after $p_1$ has length at most $|m_i a_i'| + |a_i' b_i'| + |b_i' m_{i+1}|$. We now apply Lemma~\ref{lem:SummingRectangles} to \C{p}{p_1} and all $C_i$ following $p_1$ and obtain that the total length of the path from $p$ to $q$ is at most $|p a| + |a b| + |b q|$. 
\end{proof}

Using the above lemma, we improve the upper bounds on the spanning ratio of the constrained generalized Delaunay graph, that uses an arbitrary rectangle as its empty \circle, compared to the general upper bound implied by Theorem~\ref{theo:GeneralizedDelaunaySpanningRatio}. 

\begin{lemma}
  Let $p$ and $q$ be two vertices that can see each other. Let $l$ and $s$ be the length of the long and short side of \C{p}{q}. The \delGraph contains a path between $p$ and $q$ of length at most $\left( \frac{2 l}{s} + 1 \right) \cdot \left( |p_x - q_x| + |p_y - q_y| \right)$. 
\end{lemma}
\begin{proof}
  We slightly abuse notation and let \C{p}{q} be the rectangle that is a homothet of $C$ with $p$ and $q$ on its boundary, such that $p$ lies in a corner of \C{p}{q}. We assume without loss of generality that $p$ lies on the Southwest corner and $q$ lies on the East boundary. Note that this implies that the slope of $p q$ is non-negative, i.e., $p_x < q_x$ and $p_y \leq q_y$. We prove the lemma by induction on the rank of \C{x}{y} when ordered by size, for any two visible vertices $x$ and $y$, such that $x$ lies in a corner of \C{x}{y}. In fact, we show that the \delGraph contains a path between $x$ and $y$ of length at most $c \cdot (q_x - p_x) + d \cdot (q_y - p_y)$ and derive bounds on $c$ and $d$. 

  \textbf{Base case:} If \C{p}{q} is the smallest rectangle with $p$ in a corner, then \C{p}{q} does not contain any vertices visible to both $p$ and $q$: Let $u$ be a vertex in \C{p}{q} that is visible to both $p$ and $q$. Let \C{p}{u} be the rectangle with $p$ in a corner and $u$ on its boundary. Since $u$ lies in \C{p}{q}, \C{p}{u} is smaller than \C{p}{q}, contradicting that \C{p}{q} is the smallest rectangle with $p$ in a corner. Hence, \C{p}{q} does not contain any vertices visible to both $p$ and $q$, which implies that $p q$ is an edge of the \delGraph. Hence, the \delGraph contains a path between $p$ and $q$ of length at most $|p q| \leq (q_x - p_x) + (q_y - p_y) \leq c \cdot (q_x - p_x) + d \cdot (q_y - p_y)$, provided that $c \geq 1$ and $d \geq 1$. 

  \textbf{Induction step:} We assume that the lemma holds for all rectangles \C{x}{y} smaller than \C{p}{q}, with $x$ in some corner of \C{p}{q}. If $p q$ is an edge of the \delGraph, by the triangle inequality, $|p q|$ is at most $|p_x - q_x| + |p_y - q_y|$. 

  If there is no edge between $p$ and $q$, there exists a vertex $u$ in \C{p}{q} that is visible from both $p$ and $q$. We first look at the case where $u$ lies below $p q$. Let $g$ be the intersection of the South boundary of \C{p}{q} and the line though $q$ parallel to the diagonal of \C{p}{q} through $p$, and let $h$ be the Southeast corner of \C{p}{q} (see Figure~\ref{fig:RectangleBelow}). If $u$ lies in triangle $p g q$, by induction we have that the path from $p$ to $u$ has length at most $c \cdot (u_x - p_x) + d \cdot (u_y - p_y)$ and the path from $u$ to $q$ has length at most $c \cdot (q_x - u_x) + d \cdot (q_y - u_y)$. Hence, there exists a path from $p$ to $q$ via $u$ of length at most $c \cdot (q_x - p_x) + d \cdot (q_y - p_y)$. 
  
  \begin{figure}[ht]
    \begin{center}
      \includegraphics{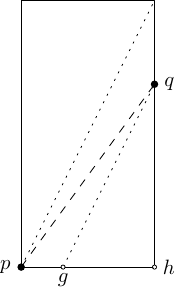}
    \end{center}
    \caption{Rectangle \C{p}{q} with points $g$ and $h$.}
    \label{fig:RectangleBelow}
  \end{figure}

  If $u$ lies in triangle $g h q$, by induction we have that the path from $p$ to $u$ has length at most $c \cdot (u_x - p_x) + d \cdot (u_y - p_y)$ and the path from $q$ to $u$ has length at most $d \cdot (q_x - u_x) + c \cdot (q_y - u_y)$. When we take $c$ and $d$ to be equal, this implies that there exists a path from $p$ to $q$ via $u$ of length at most $c \cdot (q_x - p_x) + d \cdot (q_y - p_y)$. 
  
  If there does not exist a vertex below $p q$ that is visible to both $p$ and $q$, than Lemma~\ref{lem:VisibleVertex} implies that there are no vertices in \reg{\C{p}{q}}{p}{q} below $p q$ that are visible to $p$ and that there are no vertices in \reg{\C{p}{q}}{q}{p} below $p q$ that are visible to $q$. Hence, we can apply Lemma~\ref{lem:HalfEmptyRectangle} and obtain that there exists a path between $p$ and $q$ of length at most $|p a| + |a b| + |b q|$, where $a$ and $b$ are the Northwest and Northeast corner of \C{p}{q}. Since $|a b|$ is $(q_x - p_x)$ and $|b q| \leq |p a| \leq \frac{l}{s} \cdot (q_x - p_x)$, we can upper bound $|p a| + |a b| + |b q|$ by $c \cdot (q_x - p_x)$ when $c$ is at least $\left( \frac{2 l}{s} + 1 \right)$. Hence, since $c$ and $d$ need to be equal, we obtain that all cases work out when $c = d = \left( \frac{2 l}{s} + 1 \right)$. 
\end{proof}

Finally, since $(|p_x - q_x| + |p_y - q_y|) / |p q|$ is at most $\sqrt 2$, we obtain the following theorem. 

\begin{theorem}
  The \delGraph using an empty rectangle as empty \circle has spanning ratio at most $\sqrt{2} \cdot \left( \frac{2 l}{s} + 1 \right)$.
\end{theorem}

The above theorem is quite a bit tighter than Theorem~\ref{theo:GeneralizedDelaunaySpanningRatio}, though it only holds for rectangles. For general rectangles, the dependency on the aspect ratio is lowered from cubic to linear. For squares, the implied spanning ratio drops from $24 \cdot \sqrt{4 + 2\sqrt{2}}  \approx 62.72$ to $3 \cdot \sqrt{2} \approx 4.25$, which is far closer to the tight ratio of 2.61 in the unconstrained setting~\cite{BGHP12}.

\subsection{Lower Bound for Rectangles}
In this section we provide a lower bound on the spanning ratio of \delGraph{s} using an empty rectangle as empty \circle. Like the upper bound, this lower bound is linear in the aspect ratio of the rectangle, hence the upper bound is at most a constant factor removed from the tight spanning ratio. 

\begin{theorem}
  Delaunay triangulations based on rectangles have spanning ratio at least $\sqrt{2} \cdot \sqrt{(l/s)^2 + 1 + (l/s) \cdot \sqrt{(l/s)^2 + 1}}$, where $l$ and $s$ are the length of the long and short side of the rectangle. 
\end{theorem}
\begin{proof}
Without loss of generality, we assume that $l$ is the height of the rectangle and $s$ is its width. We construct the lower bound as follows: We make two columns of $n/2$ vertices each, such that the horizontal distance between the two columns is $s$ and the height of each column is $l + y$, for $y > 0$ to be defined later. We label the vertices in the left column $l_1, l_2, ..., l_\frac{n}{2}$ and those in the right column $r_1, r_2, ..., r_\frac{n}{2}$. Next, we shift the left column up by $l - \epsilon$, for some arbitrarily small $\epsilon > 0$ (see Figure~\ref{fig:LowerBound}a), and move the vertices an arbitrarily small distance in horizontal direction, such that $l_i$ lies to the right of $l_{i+1}$ and $r_i$ lies to the right of $r_{i+1}$ for $1 \leq i < n/2$ (see Figure~\ref{fig:LowerBound}b). 

\begin{figure}[ht]
  \begin{center}
    \includegraphics{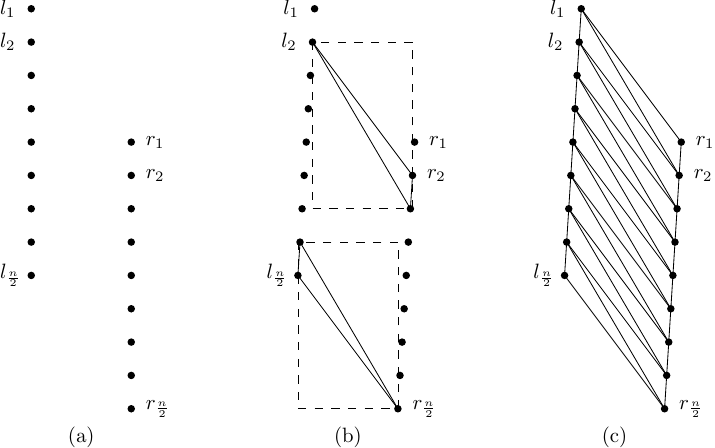}
  \end{center}
  \caption{Constructing the lower bound: (a) Placing the vertices, (b) adding the triangles for two rectangles, (c) the resulting Delaunay graph.}
  \label{fig:LowerBound}
\end{figure}

The placement of vertices guarantees that the Delaunay triangulation contains the edges $l_i l_{i+1}$, $r_i r_{i+1}$, and $l_i r_{i+1}$ for $1 \leq i < n/2$, as well as the edges $l_i r_i$  for $1 \leq i \leq n/2$. The resulting graph is shown in Figure~\ref{fig:LowerBound}c. 

We proceed to analyze the length of the shortest path between $l_\frac{n}{2}$ and $r_1$, specifically the one via $l_1$. Since all perturbations can be made arbitrarily small, this path has length $y + l + \sqrt{l^2 + s^2}$ as $n \rightarrow \infty$. The Euclidean distance between $l_\frac{n}{2}$ and $r_1$ is arbitrarily close to $\sqrt{y^2 + s^2}$. This implies that the spanning ratio is lower bounded by \[\frac{y + l + \sqrt{l^2 + s^2}}{\sqrt{y^2 + s^2}}.\] 

It remains to determine the worst case value of $y$. In order to find this, we determine the derivative of the spanning ratio with respect to $y$: \[\frac{s^2 - y \cdot (l + \sqrt{l^2 + s^2})}{(y^2 + s^2)^{3/2}}.\] This derivative is 0 when $y$ equals $\sqrt{l^2 + s^2} - l$. It is easy to verify that this is a maximum and that the spanning ratio is \[\frac{2 \sqrt{l^2 + s^2}}{\sqrt{(\sqrt{l^2 + s^2} - l)^2 + s^2}}.\] Since both $l$ and $s$ are positive, this expression can be rewritten to \[\sqrt{2} \cdot \sqrt{\frac{\sqrt{l^2 + s^2} \cdot (l + \sqrt{l^2 + s^2})}{s^2}},\] which in turn can be rewritten to \[\sqrt{2} \cdot \sqrt{\left(\frac{l}{s}\right)^2 + 1 + \left(\frac{l}{s}\right) \cdot \sqrt{\left(\frac{l}{s}\right)^2 + 1}}.\] 
\end{proof}

We note that for a square $l$ and $s$ are equal and the lower bound becomes $\sqrt{2} \cdot \sqrt{2 + \sqrt{2}} = \sqrt{4 + 2 \sqrt{2}}$, matching the lower bound by Bonichon~\etal~\cite{BGHP12}. This leads us to conjecture that the lower bound actually is the tight spanning ratio of the constrained Delaunay graphs for rectangles.

\section{Conclusion}
We showed that every \delGraph is a plane spanner, whose spanning ratio depends on the $\alpha$-diamond property and the visible-pair $\kappa$-spanner property. In the special case where the empty \circle is a rectangle, we reduce the spanning ratio by showing that it depends linearly on the aspect ratio of the rectangle used to construct the graph. 

While the results presented here are very general, the implied upper bound on the spanning ratio is likely to be far from tight. Indeed, as mentioned in Section~\ref{sec:SpanningRatioDelaunay}, proofs designed with a specific convex shape in mind give better upper bounds, some of which are even tight. Also considering the results presented for rectangles, which lower the dependency of the spanning ratio on the aspect ratio from cubic to linear, we conjecture that similar improvements can be made for other families of convex shapes. 

In light of other recent results in the constrained setting, such as the fact that Yao- and $\theta$-graphs with sufficiently many cones are spanners, the result presented in this paper raises a tantalizing question: What conditions need to hold for a graph to be a spanner in the constrained setting? In particular, these and previous results show a number of sufficient conditions, but do not immediately give rise to a set of necessary conditions.

\bibliographystyle{plain}
\bibliography{references}

\end{document}